%% file: root.tex
\documentclass[letterpaper, 10 pt, conference]{ieeeconf}  

\IEEEoverridecommandlockouts                              
\input{preamble}

\title{\LARGE \bf Partial Resilient Leader-Follower Consensus in Time-Varying Graphs
}

\author{Haejoon Lee and Dimitra Panagou
\thanks{*This work is supported by the Air Force Office of Scientific Research (AFOSR) under FA9550-23-1-0163}
\thanks{All authors are with the Robotics Department, University of Michigan, Ann Arbor, MI, USA
        {\tt \{haejoonl, dpanagou\}@umich.edu}}
}

\begin{document}

\maketitle
\thispagestyle{empty}
\pagestyle{empty}

\begin{abstract} 
This work studies resilient leader-follower consensus with a bounded number of adversaries. Existing approaches typically require robustness conditions of the entire network to guarantee resilient consensus. However, the behavior of such systems when these conditions are not fully met remains unexplored. To address this gap, we introduce the notion of partial leader-follower consensus, in which a subset of non-adversarial followers successfully tracks the leader’s reference state despite insufficient robustness. We propose a novel distributed algorithm --- the Bootstrap Percolation and Mean Subsequence Reduced (BP-MSR) algorithm --- and establish sufficient conditions for individual followers to achieve consensus via the BP-MSR algorithm in arbitrary time-varying graphs. We validate our findings through simulations, demonstrating that our method guarantees partial leader-follower consensus, even when standard resilient consensus algorithms fail.
\end{abstract}

\section{Introduction}

Consensus is a process where multiple agents agree to a common value. However, standard consensus protocols are vulnerable to adversarial agents that transmit false information to the network, potentially causing significant performance degradation or even failure. This has motivated extensive research on resilient consensus, which ensures consensus in the presence of adversarial agents~\cite{LeBlanc13,zhang2012, saldana2017, time_varying_strongly_usevitch20, CDC2025, joint_robustness2023, rezaee21, yuan_multi_hop2021,CDC2024, akgün2025multi,sood2025balancing, TAC2025}.

Many works rely on \textit{Mean Subsequence Reduced} (MSR)-type algorithms~\cite{msr_type_2012, LeBlanc13, saldana2017, yuan_multi_hop2021}. In the algorithms, each non-adversarial (normal) agent discards some extreme values from its neighbors before updating its state, ensuring that the updated value is not influenced by the adversarial agents. For many MSR-type algorithms to succeed, the communication graph is assumed to satisfy topological conditions called $\textit{r-robustness}$ or $(r,s)$-robustness, which are sufficient (and necessary) conditions for normal agents to achieve consensus within the convex hull of their initial values~\cite{LeBlanc13, zhang2012}.

A closely related problem is that of resilient leader-follower consensus, where a subset of agents (leaders) propagate a reference signal that the rest of the network (followers) aims to track, despite the presence of adversarial agents~\cite{ren2010, time_varying_strongly_usevitch20, dimar2008}. Resilient broadcasting problems have also been studied~\cite{CPA1,CPA3, zhang2012}, where algorithms such as the \textit{Certified Propagation Algorithm} (CPA) were proposed to enable a trustworthy leader to reliably broadcast information despite the presence of faulty nodes. In~\cite{leblanc_parameter_est}, the authors addressed the problem of resilient distributed estimation by leveraging reliable agents that are directly connected to others. Additionally,~\cite{mitra2019} established conditions under which information can be resiliently transmitted from a set of source nodes to other nodes without direct access to the information. 

These leader-follower results similarly depend on robustness conditions of the entire network. One of the most studied is the strong $r$-robustness~\cite{zhang2012}, which roughly quantifies the redundant information flow from a designated subset of nodes to the rest of the network. This condition was later extended to time-varying graphs~\cite{time_varying_strongly_usevitch20}. A related concept, termed $r$ leader-follower robustness, was introduced in~\cite{rezaee21} for time-invariant networks, but it assumes there exists a trustworthy leader. More recently, the notion of joint $r$-robust following was proposed in~\cite{yuan2024reaching}, providing necessary and sufficient conditions for resilient leader-follower consensus in time-varying graphs with multi-hop communication. These robustness notions have served as the foundation for numerous future studies~\cite{li_resilient_2024, wang_resilient_2024, cui_resilient_2025,shang_resilient_2023,ICRA2025}.


Despite these efforts, the system behavior under MSR algorithms in non-ideal settings --- \textbf{where these robustness conditions of the entire network fail} --- remains poorly understood. Since robustness is often hard to guarantee in energy- or communication-constrained systems, analyzing these scenarios allows us to formally quantify performance degradation and limits.
Some progress has been made for \textit{leaderless settings}. For example,\cite{community_consensus2024,shang2022cluster} study community or cluster consensus, where multiple sets of agents converge to multiple distinct values when the robustness of the entire network is not sufficient. However, they also assume certain robustness within communities (clusters) and do not study cases where these conditions fail. 
In another line of work,\cite{khalyavin2024} investigates the \textit{non-resiliency} of a graph by examining the number of possible \textit{non-convergent nodes} (agents that fail to reach consensus with any other nodes) when the network lacks sufficient robustness, but these results are limited to specific classes of graphs.

\emph{Contributions:} This paper studies the behavior of multi-agent systems under insufficient network robustness for resilient leader-follower consensus in \textit{arbitrary} time-varying graphs. Our main contributions are:
\begin{itemize}
    \item We develop a novel \textit{Bootstrap Percolation and MSR} (BP-MSR) algorithm, that guarantees partial resilient leader-follower consensus, where a subset of normal followers achieve consensus.
    \item We establish sufficient conditions for followers to achieve consensus via the BP-MSR algorithm. Unlike previous work~\cite{zhang2012,time_varying_strongly_usevitch20, rezaee21, yuan2024reaching} that relies on robustness of the entire network to guarantee convergence of \textit{all followers}, our conditions apply to each follower, focusing on convergence of \textit{individual followers}. 
    \item We validate our approach through simulations, showing that the BP-MSR algorithm guarantees partial resilient leader-follower consensus even when traditional resilient consensus algorithms fail to do so.
\end{itemize}


\section{Preliminaries}
\label{sec:prelim}
We denote the cardinality of a set $\mathcal S$ as $|\mathcal S|$. We denote the set of non-negative and positive integers as $\mathbb Z_{\geq 0}$ and $\mathbb Z_{>0}$. We denote a simple time-varying digraph as $\mathcal G[t] = (\mathcal {V},\mathcal{E}[t])$ where $\mathcal {V}$ and $\mathcal{E}[t]$ are the finite vertex set and the time-varying directed edge set, respectively. We also denote $\Gbb=(\Gcal[t])_{t\in \mathbb Z_{\geq 0}}$ as a sequence of time-varying graphs. Within the set $\mathcal V$, we have leaders $\mathcal L=\{1,\dots, l\} \subset \mathcal V$ that propagate the reference state according to the same function $f_r:\mathbb Z_{\geq 0} \to \mathbb R$ to followers $\mathcal F = \{l+1,\dots, n\} = \mathcal V \setminus \mathcal L$. We denote $|\mathcal L|=l$ and $|\mathcal F|=f$. A directed edge $(i,j)\in \mathcal{E}[t]$ indicates that agent $j$ is able to receive information from agent $i$ at time $t$. This implies agent $i$ is an \textit{in-neighbor} of agent $j$, and agent $j$ is an \textit{out-neighbor} of agent $i$. The in-neighbor and out-neighbor sets of agent $i$ are denoted as $\mathcal N_i[t] =\{j\in \mathcal V \;|\; (j,i) \in \mathcal{E}[t]\}$ and $\mathcal N_i^o[t] =\{j\in \mathcal V \;|\; (i,j) \in \mathcal{E}[t]\}$, respectively. The extended in-neighbor set is given as $\Bcal_i[t] = \Ncal_i[t]\cup \{i\}$. 

\subsection{Bootstrap Percolation} At a fixed time step $t\in \Zbb_{\geq 0}$, bootstrap percolation (BP) models the spread of \textit{activation} of nodes of $\mathcal G[t]$ from a set of initial active nodes $\mathcal L\subset \Vcal$, given a user-defined threshold $r\in \mathbb Z_{>0}$~\cite{jason2012}.
In BP, each node is either \textit{active} or \textit{inactive}. At each iteration of BP, each follower $i \in \Fcal$ becomes active if it has at least $r$ active neighbors and remains active until the process terminates, which happens after at most $f$ iterations. 

We modify BP process slightly such that each agent's \textit{update occurs locally}. 
Formally, let $q_i[t,k]\in \{0,1\}$ be the true activation state of node $i$ at time $t$ and BP iteration $k\in \{0,\dots, f\}$, where $q_i[t,k]=1$ if active and inactive otherwise. Also, let $q_i^j[t,k]$ be the activation state of agent $i$ received by agent $j$. For $k\in\{1,\dots,f\}$, agent $i$ shares its activation state $q_i[t,k]$ with its out-neighbors and updates as
\begin{equation}
q_i[t,k]=
\begin{cases}
   1 & \text{if $q_i[t,k-1]=1$} \text{ or } \text{$\sum_{j \in \Ncal_i[t]} q_j^i[t,k-1]\geq r$},\\
   0 & \text{otherwise},
\end{cases}
\label{eq:percolation}
\end{equation}where $q_i[t,0]=1$ $\forall i \in \Lcal$ and $q_i[t,0]=0$ $\forall i \in \Fcal$. 
For brevity, we denote $q_i[t]:=q_i[t,f]$.

\subsection{Resilient Consensus}
We review resilient leader-follower consensus when agents cannot distinguish between leaders and followers.
Let $x_i[t]\in \mathbb R$ denote the consensus state of agent $i$ at time $t$, and $x_i^j[t]$ the copy received by agent $j$. 
Each agent $i$ shares $x_i[t]$ with its out-neighbors and updates its state using
\begin{equation}
x_i[t+1]=
    \begin{cases}
        f_r[t], & i\in \mathcal L, \\
     \sum_{j\in \mathcal B_i[t]} w_{ij}[t]x_j^i[t], & i \in \mathcal F,
    \end{cases}
    \label{eq:linear}
\end{equation}
where $w_{ij}[t]$ is the weight assigned to $x_j^i[t]$ by agent $i$. We assume $\exists \alpha \in(0,1)$ such that $\forall i \in \mathcal V$,
\begin{itemize}
    \item $w_{ij}[t]\geq \alpha$ if $j\in \Bcal_i[t]$, or $w_{ij}[t]=0$ otherwise,
    \item $\sum_{j=1}^n w_{ij}[t]=1$.
\end{itemize}
However, as shown in~\cite{strongly_usevitch18}, the protocol~\eqref{eq:linear} does not guarantee consensus with adversaries, which we define below:

\begin{definition}[\textbf{Adversarial agent}]
    \label{def:byzantine}
    An agent $a\in \mathcal V$ is an \textbf {adversary} if it transmits consensus state to all out-neighbors at all $t\geq \Zbb_{\geq 0}$ but does not follow the  update protocols~\eqref{eq:percolation} and/or \eqref{eq:linear}, or if it
does not send the same values to all out-neighbors at some time $t$.
\end{definition}

 The adversarial agents in this paper are a slight modification of the well-known Byzantine agents~\cite{LeBlanc13, time_varying_strongly_usevitch20}: each $a\in \Acal$ may manipulate both its consensus and activation states. Both followers and leaders can be adversaries. Non-adversarial agents are called normal agents. We denote $\Acal$ as the set of adversaries and $\Mcal = \Vcal \setminus \Acal$ as the set of normal agents, with $\Fcal_{\Mcal} = \Fcal \cap \Mcal\neq \emptyset$ and $\Lcal_{\Mcal} = \Lcal \cap \Mcal\neq \emptyset$.


\begin{assumption} For all $a \in \Acal$, $u \in \Ncal_a^o[t]$, $k \in \{0,\dots, f\}$, and $t \in \mathbb{Z}_{\geq 0}$, we have $q_a^u[t,k] \in \{0,1\}$. 
    \label{assum:misreport}
\end{assumption}
Hence, an adversary cannot transmit non-binary values during the BP process. This is a direct violation of~\eqref{eq:percolation} and would lead to detections, which adversaries would try to avoid. In addition, we consider the following scope of threat:

\begin{definition}[$\mathbf F$\textbf{-local}]
    \label{def:flocal}
    A set $\mathcal S \subset \mathcal V$ is \textbf{$\mathbf F$-local} if all other nodes have at most $F$ nodes of $\mathcal S$ as their in-neighbors (i.e. $|\mathcal N_i[t]\cap \mathcal S|\leq F$, $\forall i \in \mathcal V \setminus \mathcal S$) $\forall t\geq \Zbb_{\geq 0}$.
\end{definition}

In response to adversaries, works like ~\cite{strongly_usevitch18, time_varying_strongly_usevitch20, yuan2024reaching,rezaee21} studied how to achieve full resilient leader-follower consensus, whose definition is as follows:
\begin{definition}
The followers in $\Fcal_{\Mcal}$ are said to achieve \textbf{full resilient leader-follower consensus} if, for every initial consensus state $x_i[0]$, $i \in \Fcal_{\Mcal}$, and for every shared consensus state of adversaries $x_a^u[t]$, $a \in \Acal$, $u\in \Ncal_a^o[t]$, $t \in \Zbb_{\geq 0}$,
\eqn{\lim_{t\to \infty}\norm{x_i[t]-f_r[t]}=0, \quad \forall i \in \Fcal_{\Mcal}.}
\end{definition}
\noindent Their results, however, hold only under certain topological conditions specified below.

\begin{definition}[$\mathbf r$-\textbf{reachable} \cite{LeBlanc13}]
    \label{reachability}
    Let $\mathcal G[t] = (\mathcal V,\mathcal{E}[t])$ be a graph at time $t$ and $\mathcal S$ be a nonempty subset of $\mathcal V$. The subset $\mathcal S$ is $\mathbf r$-\textbf{reachable} at $t$ if $\exists i\in \mathcal S$ such that $|\mathcal N_i[t] \backslash \mathcal S|\geq r$.
\end{definition}

\begin{definition}[\textbf{strongly} $\mathbf r$-\textbf{robust} \cite{zhang2012}]
    \label{robustness}
    A graph $\mathcal G[t] = (\mathcal V,\mathcal{E}[t])$ is \textbf{strongly} $\mathbf r$-\textbf{robust} with respect to $\mathcal S_1 \subset \mathcal V$ at time $t$ if $\forall \mathcal S_2 \subset \mathcal V \setminus \mathcal S_1$ such that $\mathcal S_2 \neq \emptyset$, $\mathcal S_2$ is $r$-reachable at $t$.
\end{definition}
If a network is strongly $(2F+1)$-robust with respect to $\Lcal$ $\forall t\in \mathbb{Z}_{\geq 0}$, then followers in $\mathcal{F}_{\mathcal{M}}$ achieve full resilient leader-follower consensus via the W-MSR algorithm~\cite{strongly_usevitch18}. This was later extended to incorporate a union of time-varying graphs over a time horizon in~\cite{time_varying_strongly_usevitch20}. Similar conditions appear in~\cite{yuan2024reaching, rezaee21}. However, in contrast to those conditions, there exists a useful result in regard to strong $r$-robustness:

\begin{lemma}[{\cite[Corollary 1 (revised)]{ICRA2025}}]
    Let $\mathcal G[t]$ be a graph at time $t$, with threshold $r \ge 2$ and initial set $\Lcal \subseteq \Vcal$. Suppose that all agents $i\in \Vcal$ truthfully share and follow~\eqref{eq:percolation} to update $q_i[t,k]$ for all $k\in[0,f]$. Then, $\mathcal G[t]$ is strongly $r$-robust with respect to $\Lcal$ at time $t$ if and only if $q_i[t] = 1$ $\forall i \in \Fcal$.
\label{lem:strong_r_not_mine}
\end{lemma}
\begin{remark}
\label{remark:on_lemma}
    Note \Cref{lem:strong_r_not_mine} holds only if all agents in $\mathcal G[t]$ truthfully update and share their activation states throughout BP process. However, in the presence of adversaries, this assumption no longer holds. Specifically, adversaries $a \in \Acal$ can share wrong values of $q_a[t, k]$ at any $t \in \mathbb{Z}_{\geq 0}$ and $k \in \{0,\dots, f\}$. As a result, a normal follower $i \in \Fcal_{\Mcal}$ may remain inactive (i.e., $q_i[t, k] = 0$) at $t$, even if it would have become active when all agents follow~\eqref{eq:percolation}, and vice versa.
\end{remark}


\section{Problem Statement}
Much of the existing work on resilient leader-follower consensus~\cite{time_varying_strongly_usevitch20, yuan2024reaching, rezaee21} relies on the assumption that the entire network always satisfies certain topological robustness conditions. In general, these conditions require dense network structures which might be difficult to maintain in practice. Therefore, our goal is to study the system’s behavior over an arbitrary sequence of graphs, even when robustness conditions are not met.

To this end, rather than focusing on full resilient leader-follower consensus, where all normal followers are required to converge, we introduce the notion of partial resilient leader-follower consensus:
\begin{definition}
Let $\Fcal_{\Ccal} \subset \Fcal_{\Mcal}$ be a non-empty subset of normal followers.  
The followers in $\Fcal_{\Ccal}$ are said to achieve \textbf{partial resilient leader-follower consensus} if, for every initial consensus state $x_i[0]$, $i \in \Fcal_{\Mcal}$, and for every shared consensus state of adversaries $x_a^u[t]$, $a \in \Acal$, $u\in \Ncal_a^o[t]$, $t \in \Zbb_{\geq 0}$,
\eqn{\lim_{t\to \infty}\norm{x_i[t]-f_r[t]}=0, \quad \forall i \in \Fcal_{\Ccal}.}
We call $\Fcal_{\Ccal}$ the \textbf{convergent set} and its members \textbf{convergent followers}. The set $\Fcal_{\Ccal}' := \Fcal_{\Mcal} \setminus \Fcal_{\Ccal}$ is the \textbf{non-convergent set}, with members called \textbf{non-convergent followers}.
\end{definition}
\begin{remark}
    In~\cite{khalyavin2024}, non-convergent nodes are defined in the context of leaderless consensus as nodes that may fail to converge to a common value with any other nodes. In contrast, we define non-convergent followers as followers that may fail to converge to \textit{leader's reference state}.
\end{remark}


Although adversaries may manipulate both their consensus and activation states, our focus is on the convergence of the \emph{consensus states}. That is, we do not require the activation states received by a normal agent to be correct. Nevertheless, as we demonstrate in later sections, how adversaries share their activation states influences which followers ultimately achieve convergence (see~\Cref{sec:convergent_set}).

 We now formally state our problem:
\begin{problem}
    Let $\Gbb=(\Gcal[t])_{t\in \mathbb Z_{\geq 0}}$ be an arbitrary sequence of time-varying digraphs under an $F$-local $\Acal$. We aim to characterize sufficient conditions under which followers in $\Fcal_{\Ccal}\subseteq \Fcal_{\Mcal}$ achieve partial resilient leader–follower consensus.
\end{problem}


\section{Partial Resilient Leader-Follower Consensus}
Consider a network $\Gcal[t]=(\Vcal, \Ecal[t])$ that fails to be strongly $r$-robust at time $t$. Since robustness conditions are defined over the entire network topology, that implies there exists a nonempty subset $\Fcal_{\Wcal}[t] \subset \Fcal$ whose members lack sufficient connections to satisfy the robustness conditions. However, the remaining followers in $\Fcal_{\Rcal}[t]=\Fcal \setminus \Fcal_{\Wcal}[t]$ (if nonempty) are indeed sufficiently connected. In other words, the \textit{subgraph} induced by a set $\Lcal\cup \Fcal_{\Rcal}[t]\subseteq \Vcal$ is strongly $r$-robust with respect to $\Lcal$. Using this observation, we propose a novel distributed algorithm, which we refer to as the BP-MSR algorithm (\Cref{alg:msr}). 

The key idea of the BP-MSR algorithm is that agents engage in a consensus update \textit{only at time $t$ when they can locally verify that they belong to strongly $r$-robust subgraphs.} 
For each $t\in \Zbb_{\geq 0}$, all agents $i\in \Mcal$ first compute their activation states $q_i[t]$ synchronously via BP~\eqref{eq:percolation} with $r=2F+1$ (lines 3-8). Then, leaders $i\in \Lcal$ transmit their states and update them using $f_r$ (lines 9-11), while followers $i\in \Fcal$ use $q_i[t]$ to locally decide whether to engage in consensus protocol synchronously (lines 12-21). If $q_i[t]=1$, i.e., followers $i$ belong to a strongly $(2F+1)$-robust subgraph with respect to $\Lcal\cup \Acal$ at time $t$ (shown in~\Cref{lem:partial_2f+1}), they transmit $x_i[t]$ and update $x_i[t+1]$ via the W-MSR algorithm; otherwise, they remain inactive for that time.

\begin{remark}
    The BP-MSR algorithm forces the followers $i\in \Fcal_{\Mcal}$ to send and update their states at time $t\in \Zbb_{\geq 0}$ \textit{only when they have sufficient connections to achieve resilient leader-follower consensus} (which is encoded as $q_i[t] = 1$). This selective participation enables a subset of followers to achieve partial resilient consensus while preventing non-convergent or adversarial values from propagating through the network. In fact, when the entire graph $\Gcal[t]$ is sufficiently connected, i.e., $q_i[t]=1$ $\forall i \in \Fcal_{\Mcal}$, the algorithm reduces to the standard W-MSR algorithm at time $t$. 
\end{remark}

\subsection{Analysis of the BP-MSR Algorithm}
Now, we present the analysis for our algorithm.

\begin{lemma}
Let~\Cref{assum:misreport} hold and $\Gcal[t]=(\Vcal, \Ecal[t])$ be a time-varying digraph. Suppose that each agent $i\in \Mcal$ computes $q_i[t]$ via~\eqref{eq:percolation} with $r=2F+1$. Let $\Fcal_{\Rcal}[t] =  \{i \in \Fcal_{\Mcal} \mid q_i[t]=1\}$, and also let $\Gcal_{\Rcal}[t] = (\Vcal_{\Rcal}[t], \Ecal_{\Rcal}[t])$ be a subgraph induced by $\Vcal_{\Rcal}[t]= \Acal \cup \Lcal\cup \Fcal_{\Rcal}[t]$. If $|\Fcal_{\Rcal}[t]|> 0$, $\Gcal_{\Rcal}[t]$ is strongly $(2F+1)$-robust with respect to $\Lcal\cup \Acal$ at time $t$. 
    \label{lem:partial_2f+1}
\end{lemma}
\begin{proof}
By definition, for all $i\in\Fcal_{\Rcal}[t]$ there exists $\bar {k}\in\{0,\dots,f-1\}$ such that the following inequality holds: $\sum_{j \in \Ncal_i[t]\setminus \Acal} q_j^i[t,\bar{k}] + \sum_{a \in \Acal} q_a^i[t,\bar{k}] \geq 2F+1$. Because $q_a^i[t,k]\in \{0,1\}$ for all $a \in \Acal$, $i \in \Mcal$, and $k\in\{0,\dots, f\}$ by~\Cref{assum:misreport}, setting $q_a^i[t,k]=1$ for all $a\in \Acal$, $i\in \Mcal$, and $k\in\{0,\dots,f-1\}$ would still ensure that the inequality holds for the same $\bar{k}$. Note setting $q_a^i[t,k]=1$ for all $i \in \Mcal$ and $k\in\{0,\dots,f-1\}$ is equivalent to making an adversary $a\in \Acal$ a leader at time $t$. Then, applying~\Cref{lem:strong_r_not_mine} with the threshold $r=2F+1$ and initial set $\Lcal\cup \Acal$, $\Gcal_{\Rcal}[t]$ is strongly $(2F+1)$-robust with respect to $\Lcal\cup \Acal$.
\end{proof}


\begin{algorithm}[h]
\caption{BP-MSR Algorithm}
\label{alg:msr}
\begin{algorithmic}[1]
\State \textbf{Input:} Parameter $F$
\State \textbf{Output:} $x_i[t+1]$
\State Initialize activation: $q_i[t,0]\gets 1$ if $i\in \Lcal$, else $0$

\For{$k\in\{1,\dots, f\}$}
\State Transmit $q_i[t,k-1]$ to its out-neighbors $u \in \Ncal_i^o[t]$
\State Update $q_i[t,k]$ through~\eqref{eq:percolation} with $r=2F+1$
\EndFor
\State $q_i[t] \gets q_i[t,f]$

\If {$i \in \Lcal$}
\State Transmit $x_i[t]$ to $u\in \Ncal_i^o[t]$
\State $x_i[t+1] \gets f_r[t]$
\ElsIf{$i \in \Fcal$ and $q_i[t]$ is $1$}
\State Transmit $x_i[t]$ to $u\in \Ncal_i^o[t]$
\State $\Ical_i[t] \gets \{j \in \Ncal_i[t] \mid \text{agent } i \text{ received } x_j^i[t]\}\cup\{i\}$
\State $\Hcal_i[t]\gets \{x_j^i[t] \mid j \in \Ical_i[t] \text{ and }x_j^i[t] > x_i[t]\}$
\State $\Ocal_i[t]\gets \{x_j^i[t] \mid j \in \Ical_i[t] \text{ and } x_j^i[t] < x_i[t]\}$
\State Sort $\Hcal_i[t]$ and $\Ocal_i[t]$; discard up to $F$ largest and 

smallest values in $\Hcal_i[t]$ and $\Ocal_i[t]$, respectively
\State Let $\Rcal_i[t]$ be the set of agents whose values are 

discarded. Then, it applies
    \eqn{x_i[t+1]=\sum_{j\in \Ical_i[t]\setminus \Rcal_i[t]} w_{ij}[t]x_j^i[t],}

    where $w_{ij}[t]$ satisfies the same conditions in~\eqref{eq:linear}, 
    
    but with $\Bcal_i[t]$ replaced by $\Ical_i[t]\setminus \Rcal_i[t]$
\Else
\State $x_i[t+1]=x_i[t]$
\EndIf
\end{algorithmic}
\end{algorithm}

The lemma shows that, through lines 3-8 of the BP-MSR algorithm, each follower $i \in \Fcal_{\Mcal}$ can use $q_i[t]$ to locally determine whether it belongs to a strongly $(2F+1)$-robust subgraph with respect to $\Lcal\cup \Acal$ at time $t$. This local verification step enables each follower to decide whether to participate in the consensus in later parts of the algorithm (lines 12-21). Now, let $m[t] = \min_{i \in \Mcal} x_i[t]$ and $M[t] = \max_{i \in \Mcal} x_i[t]$ be the minimum and maximum states of normal agents at time $t$, respectively. Then, we have:

\begin{theorem}
Let~\Cref{assum:misreport} hold. Let $\Gbb=(\Gcal[t])_{t\in\mathbb Z_{\geq 0}}$ be a sequence of time-varying digraphs with an $F$-local adversary set $\Acal$, and let $f_r[t]=C_r$ $\forall t\geq t_C$. If each $i\in \Mcal$ runs the BP-MSR algorithm with parameter $F$ for all $t\geq t_C$, 

\begin{itemize}
    \item $x_i[t] \in [m[t_C], M[t_C]], \ \forall i \in \Fcal_{\Mcal}, \ \forall t \geq t_C$.
    \item Let $\Fcal_{\Ccal} =  \{i \in \Fcal_{\Mcal} \mid \lim\sup_{t\to \infty} q_i[t]=1\}$. If $|\Fcal_{\Ccal}|>0$, followers in $\Fcal_{\Ccal}$ will achieve partial resilient leader-follower consensus.
\end{itemize}

\label{thm:main}
\end{theorem}
\begin{proof}

\textbf{First Claim:} Let $\Fcal_{\Rcal}[t] =  \{i \in \Fcal_{\Mcal} \mid q_i[t]=1\}$. Note at each time $t\geq t_C$, only $i\in \Fcal_{\Rcal}[t]$ updates its value $x_i[t+1]$ with their in-neighbors' states. If $i \notin \Fcal_{\Rcal}[t]$, then $x_i[t+1]=x_i[t]\in [m[t], M[t]]$. Now consider the case where $|\Fcal_{\Rcal}[t]|\neq 0$ and $i\in \Fcal_{\Rcal}[t]$. Let $\Gcal_{\Rcal}[t] = (\Vcal_{\Rcal}[t], \Ecal_{\Rcal}[t])$ be a subgraph of $\Gcal[t]$ induced by $\Vcal_{\Rcal}[t]= \Acal \cup \Lcal \cup \Fcal_{\Rcal}[t]$, which is strongly $(2F+1)$-robust with respect to $\Lcal\cup \Acal$ by~\Cref{lem:partial_2f+1}. Also, because follower $i \in \Fcal_{\Rcal}[t]$ has $|\Ncal_i[t]\cap \Acal|\leq F$ by the definition of $F$-local model, $|\Ncal_i[t]\cap\Fcal_{\Rcal}[t]\setminus \Acal|\geq F+1$. Then if an adversary $a\in \Acal$ has a value $x_a^i[t]> M[t]$ or $x_a^i[t]<m[t]$, it will be filtered out, since it is one of the highest or lowest values in $\Hcal_i[t]$ or $\Ocal_i[t]$. Hence, all values outside the interval $[m[t], M[t]]$ will be excluded from the update. Then, follower $i$ performs a convex combination of values within $[m[t], M[t]]$ to update $x_i[t+1]$, which implies $x_i[t+1] \in [m[t], M[t]]$. Lastly, $x_i[t]=C_r \in [m[t], M[t]]$ for all $i \in \Lcal_\Mcal$ and $t\geq t_C$. Since $x_i[t+1] \in  [m[t], M[t]]$ for all $i\in \Mcal$ and $t\geq t_C$, $x_i[t+1] \in [m[t+1], M[t+1]]\subseteq [m[t], M[t]] \subseteq \cdots \subseteq [m[t_C], M[t_C]]$.

\textbf{Second Claim:} [Monotonicity of Convergent Followers] 
 Let $\Fcal_{\Ccal}' = \Fcal_{\Mcal} \setminus \Fcal_{\Ccal}$ and $\Mcal_{\Ccal} = \Lcal_\Mcal \cup \Fcal_{\Ccal}$. Then, $\forall i \in \Fcal_{\Ccal}' \  \exists \bar t_i\geq t_C$ such that $q_i[t]=0 \ \forall t \geq \bar t_i$. Now we define $\overline m[t]=\min_{i \in \Mcal_{\Ccal}} x_i[t]$ and $\overline M[t]=\max_{i \in \Mcal_{\Ccal}} x_i[t]$ as the minimum and maximum states of normal agents in $\Mcal_{\Ccal}$ at time $t$. We first show that there exists a finite time $\bar t \geq t_C$ such that $x_i[t] \in [\overline m[\bar t],\overline M[\bar t]]$ for all $ i \in \Fcal_{\Ccal}$ and $\forall t \geq \bar t$. Let $\bar t = \max_{i\in \Fcal_{\Ccal}'} t_i<\infty$ be the time where $q_i[t]=0$ $\forall i \in \Fcal_{\Ccal}'$ $\forall t \geq \bar t$. By definition, all followers $i \in \Fcal_{\Ccal}'$ do not transmit their states to out-neighbors $\forall t\geq \bar t$. Thus, removing nodes in $\Fcal_{\Ccal}'$ from the graph $\Gcal[t]$ for all $t\geq \bar t$ does not affect our analysis on the states of $\Fcal_{\Ccal}$. Then, by the same reasoning as in the proof for the first claim, we have $x_i[t] \in [\overline m[t],\overline M[t]]\subseteq \cdots \subseteq [\overline m[\bar t],\overline M[\bar t]]$ for all $i \in \Fcal_{\Ccal}$ and $t \geq \bar t$.

[Partial Leader-Follower Consensus]
We define 
\eqnN{\Xcal_M(t, \overline{\epsilon}) &= \{i \in \Fcal_{\Ccal} \mid x_i[t]> \overline M[t] - \overline{\epsilon}\}, \\
\Xcal_m(t, \underline{\epsilon})& = \{i \in  \Fcal_{\Ccal}  \mid x_i[t]< \overline m[t] + \underline{\epsilon}\},\\
\Scal_X(t, \underline{\epsilon},\overline{\epsilon}) & = \Xcal_M(t,\overline{\epsilon})\cup \Xcal_m(t, \underline{\epsilon}).
}
The set $\Scal_X(t, \underline{\epsilon},\overline{\epsilon})$ denotes a set of convergent followers whose states are outside of $[\overline{m}[t] + \underline{\epsilon}, \overline{M}[t] - \bar{\epsilon}]$. For the rest of the proof, we show that $|\Scal_X(t, \underline{\epsilon},\overline{\epsilon})|$ decreases as $t\to \infty$ with some suitable values of $\underline{\epsilon},\overline{\epsilon}\geq0$. Then, we show $\overline{M}[t]- \overline{m}[t]\to 0$ as $t\to \infty$.

 Recall that $\Gcal_{\Rcal}[t] = (\Vcal_{\Rcal}[t], \Ecal_{\Rcal}[t])$ is the subgraph induced by $\Vcal_{\Rcal}[t]= \Lcal \cup \Acal \cup \Fcal_{\Rcal}[t]$, where $\Fcal_{\Rcal}[t] =  \{i \in \Fcal_{\Mcal} \mid q_i[t]=1\}$. If $|\Fcal_{\Rcal}[t]|>0$, by~\Cref{lem:partial_2f+1}, $\Gcal_{\Rcal}[t]$ is strongly $(2F+1)$-robust with respect to $\Lcal\cup \Acal$. For all $t\geq \bar t$ such that $|\Fcal_{\Rcal}[t]|>0$, $\Fcal_{\Ccal}\cap \Fcal_{\Rcal}[t]\neq \emptyset$ and $\Fcal_{\Ccal}'\cap \Fcal_{\Rcal}[t]=\emptyset$. Let $t'\geq \bar t$ be the earliest time where $|\Fcal_{\Rcal}[t']|>0$. Note that $t'<\infty$ as $|\Fcal_{\Ccal}|>0$. We define $\underline{\epsilon} = C_r - \overline m[t']$ and $\overline{\epsilon} = \overline M[t'] - C_r$.

 We first define a set of all followers who are adjacent to at least $F+1$ normal leaders in the sequence $\Gbb$ at any time $t\geq t'$: $\Fcal_{\Lcal} = \{i_1 \in \Fcal_{\Ccal} \mid \exists t_1\in \mathbb{Z}_{\geq 0}\ \text{ s.t. } i_1 \in \Fcal_{\Rcal}[t'+t_1], \ |\Ncal_{i_1}[t'+t_1]\cap \Lcal_{\Mcal}|\geq F+1\}$. By definitions of strong $(2F+1)$-robustness and $F$-local model, for all $t\geq t_C$ such that $|\Fcal_{\Rcal}[t]|>0$, there must exist $i_1 \in \Fcal_{\Rcal}[t]\cap \Fcal_{\Lcal}$. Since $\Fcal_{\Ccal}$ is finite and $\Fcal_{\Lcal} \subseteq \Fcal_{\Ccal}$, there must exist a finite $T_1 \geq 1$ such that $\Fcal_{\Lcal} \subseteq \bigcup_{t_1=0}^{T_1-1} \Fcal_{\Rcal}[t'+t_1]$. Therefore, for any $t_1\in [0, T_1-1]$ where $|\Fcal_{\Rcal}[t'+t_1]|>0$, 
$\Scal_{1,t_1}=\{i_1 \in \Fcal_{\Ccal} \mid \Ncal_{i_1}[t'+t_1]\cap \Lcal_{\Mcal}|\geq F+1\}$ is nonempty. Then, through the BP-MSR algorithm, follower $i_1$ will use at least one $C_r$ to update $x_{i_1}[t'+t_1+1]$. 

Due to the monotonicity of the convergent agents' states, we know $x_i[t]\in [\overline m[t'],\overline M[t']]$ for all $i\in \Fcal_\Ccal$ and $t\geq t'$. To bound agent $i_1$'s state at time $t'+t_1+1$, we assume the minimum weight possible $\alpha$ is put on $C_r$, and the maximum weight is put on $\overline m[t']$ and $\overline M[t']$. Then, we get
\eqn{x_{i_1}[t'+t_1+1] & \geq \alpha C_r + (1-\alpha) \overline{m}[t'] \geq \overline{m}[t']  + \alpha\underline{\epsilon},
}
\eqn{x_{i_1}[t'+t_1+1] & \leq \alpha  C_r + (1-\alpha) \overline{M}[t'] \leq \overline{M}[t']  - \alpha \overline{\epsilon}.
}
Extending these bounds to time $t'+T_1$, we get:
\eqnN{x_{i_1}[t'+t_1+2] & \geq \alpha x_{i_1}[t'+t_1+1] + (1-\alpha) \overline{m}[t']  \\ 
& \geq \overline{m}[t']  + \alpha^2 \underline{\epsilon}, \\
x_{i_1}[t'+t_1+3] & \geq \alpha x_{i_1}[t'+t_1+ 2] + (1-\alpha) \overline{m}[t'] \\ 
& \geq \overline{m}[t']  + \alpha^3 \underline{\epsilon}, \\
 & \vdots \\
 x_{i_1}[t'+t_1+k] & \geq \alpha x_{i_1}[t'+t_1+k-1] + (1-\alpha) \overline{m}[t'] \\ 
& \geq \overline{m}[t']  + \alpha^{k} \underline{\epsilon},
}
and 
\eqnN{x_{i_1}[t'+t_1+2] & \leq \alpha x_{i_1}[t'+t_1+1] + (1-\alpha) \overline{M}[t'] \\ 
& \leq \overline{M}[t']  - \alpha^2 \overline{\epsilon},\\
x_{i_1}[t'+t_1+3] & \leq \alpha x_{i_1}[t'+t_1+2] + (1-\alpha) \overline{M}[t'] \\ 
& \leq \overline{M}[t']  - \alpha^3 \overline{\epsilon},\\
 & \vdots  \\
 x_{i_1}[t'+t_1+k] & \leq \alpha x_{i_1}[t'+t_1+k-1] + (1-\alpha) \overline{M}[t'] \\ 
& \leq \overline{M}[t']  - \alpha^{k} \overline{\epsilon}.
}
This holds for $0\leq t_1< t_1+k\leq T_1$. Let $\Kcal_1 =\{i_1 \in \Fcal_{\Ccal}\mid x_{i_1}[t'+T_1] \in [\overline m[t']+\alpha^{T_1}\underline{\epsilon}, \overline M[t']-\alpha^{T_1}\overline{\epsilon}]\}$. Since $\alpha \in [0,1]$, we have $x_{i_1}[t'+T_1] \in [\overline m[t'] +\alpha^{T_1} \underline{\epsilon}, \overline M[t'] -\alpha^{T_1} \overline{\epsilon}] \ \forall i_1 \in \Scal_1 = \cup_{t_1=0}^{T_1-1} \Scal_{1,t_1}$. Then, because (i) $\Scal_1\neq \emptyset$ , (ii) $\Scal_1 \subseteq \Kcal_1$, and (iii) $\Scal_1 \subseteq \Fcal_{\Ccal}$, $|\Scal_X(t'+T_1, \alpha^{T_1}\underline{\epsilon},\alpha^{T_1}\overline{\epsilon})|< |\Fcal_{\Ccal}|$. 

We now show that there exists a finite $T_2 \geq T_1+1$ such that $|\Scal_X(t'+T_2, \alpha^{T_2}\underline{\epsilon},\alpha^{T_2}\overline{\epsilon})|<|\Scal_X(t'+T_1, \alpha^{T_1}\underline{\epsilon},\alpha^{T_1}\overline{\epsilon})|$. Let $\Kcal_2 =\{i_2 \in \Fcal_{\Ccal}\mid x_{i_2}[t'+T_2] \in [\overline m[t']+\alpha^{T_2}\underline{\epsilon}, \overline M[t']-\alpha^{T_2}\overline{\epsilon}]\}$. Let $\Dcal_{i_2}[t'+t_2] = \Ncal_{i_2}[t'+t_2]\cap\Fcal_{\Rcal}[t'+t_2] \setminus \Scal_X(t'+t_2, \alpha^{t_2}\underline{\epsilon},\alpha^{t_2}\overline{\epsilon})$ for any $t_2 \in [T_1,T_2-1]$. First, we characterize $T_2$: let $T_2\geq T_1+1$ be the earliest time such that $\Kcal_1 \subseteq \Kcal_2$. Because all $i \in \Fcal_{\Mcal}$ always use their own states to update, the bounds on the states of followers $i_1 \in \Kcal_1$ for $t \in [t'+t_2+1, T_2]$ are:
\eqn{x_{i_1}[t'+t_2+1] & \geq \alpha x_{i_1}[t'+t_2] + (1-\alpha) \overline{m}[t'] \nonumber \\ 
& \geq \overline{m}[t']  + \alpha^{t_2+1} \underline{\epsilon}, \nonumber\\
x_{i_1}[t'+t_2+2] & \geq \alpha x_{i_1}[t'+t_2+ 1] + (1-\alpha) \overline{m}[t'] \nonumber\\ 
& \geq \overline{m}[t']  + \alpha^{t_2+2} \underline{\epsilon}, \nonumber\\
 & \vdots \nonumber\\
 x_{i_1}[t'+t_2+k] & \geq \alpha x_{i_1}[t'+t_2+k-1] + (1-\alpha) \overline{m}[t'] \nonumber\\ 
& \geq \overline{m}[t']  + \alpha^{t_2+k} \underline{\epsilon},
\label{eq:lower_bounds}
}
and 
\eqn{x_{i_1}[t'+t_2+1] & \leq \alpha x_{i_1}[t'+t_2] + (1-\alpha) \overline{M}[t'] \nonumber\\ 
& \leq \overline{M}[t']  - \alpha^{t_2+1} \overline{\epsilon},\nonumber\\
x_{i_1}[t'+t_2+2] & \leq \alpha x_{i_1}[t'+t_2+1] + (1-\alpha) \overline{M}[t'] \nonumber\\ 
& \leq \overline{M}[t']  - \alpha^{t_2+2} \overline{\epsilon},\nonumber\\
 & \vdots \nonumber \\
 x_{i_1}[t'+t_2+k] & \leq \alpha x_{i_1}[t'+t_1+k-1] + (1-\alpha) \overline{M}[t'] \nonumber\\ 
& \leq \overline{M}[t']  - \alpha^{t_2+k} \overline{\epsilon}.
\label{eq:upper_bounds}
}
Note the bound holds for $T_1\leq t_2<t_2+k\leq T_2$. Then, we know $x_{i_1}[t'+T_2] \in [\overline m[t']+\alpha^{T_2}\underline{\epsilon}, \overline M[t']-\alpha^{T_2}\overline{\epsilon}]$ $\forall i_1 \in \Kcal_1$, and thus $\Kcal_1 \subseteq \Kcal_2$. Since all $i_1\in \Kcal_1 \subseteq \Fcal_{\Ccal}$ must belong to $\Fcal_{\Rcal}[t'+t_2]$ at some $t_2 \in [T_1, T_2-1]$, such a finite $T_2$ exists.

Furthermore, by definition, $\Gcal_{\Rcal}[t]$ is strongly $(2F+1)$-robust with respect to $\Lcal\cup \Acal$ by~\Cref{lem:partial_2f+1} at $t$ if $|\Fcal_{\Rcal}[t]|>0$. Therefore, for any $t_2\in [T_1, T_2-1]$ where $|\Fcal_{\Rcal}[t'+t_2]|>0$, 
$\Scal_{2,t_2}=\{i_2 \in \Scal_X(t'+t_2, \alpha^{t_2}\underline{\epsilon},\alpha^{t_2}\overline{\epsilon}) \mid |\Dcal_{i_2}[t'+t_2]|\geq 2F+1\}$ is nonempty if $\Scal_X(t'+t_2, \alpha^{t_2}\underline{\epsilon},\alpha^{t_2}\overline{\epsilon})\neq \emptyset$. If it is nonempty, since $\Acal$ is $F$-local, $\Dcal_{i_2}[t'+t_2]$ contains at least $F+1$ normal agents $\forall i_2 \in \Scal_{2,t_2}$. Then, by definition of $\Scal_X(t'+t_2, \alpha^{t_2}\underline{\epsilon},\alpha^{t_2}\overline{\epsilon})$, $x_{i_2}[t'+t_2]>x_j^{i_2}[t'+t_2]$ or $x_{i_2}[t'+t_2]<x_j^{i_2}[t'+t_2]$ $\forall j \in \Dcal_{i_2}[t'+t_2]$. This means $i_2$ will use at least one in-neighbor's state within $[\overline{m}[t']+\alpha^{t_2}\underline{\epsilon},\overline{M}[t']-\alpha^{t_2}\overline{\epsilon} ]$ to update $x_{i_2}[t'+t_2+1]$, whose bounds are $x_{i_2}[t'+t_2+1]\in [\overline m[t'] +\alpha^{t_2+1} \underline{\epsilon}, \overline M[t'] -\alpha^{t_2+1} \overline{\epsilon}]$ by~\eqref{eq:lower_bounds} and~\eqref{eq:upper_bounds}, which is also contained by $[\overline m[t'] +\alpha^{T_2} \underline{\epsilon}, \overline M[t'] -\alpha^{T_2} \overline{\epsilon}]$. 
Now denote $\Scal_2 = \cup_{t_2=T_1}^{T_2-1}  \Scal_{2,t_2}$. Because (i) $\Scal_2 \subseteq \Kcal_2$, (ii) $\Kcal_1 \subseteq \Kcal_2$, (iii) $\Scal_{2}\not\subseteq \Kcal_1$, and (iv) $i_2 \notin \Scal_X(t'+T_2, \alpha^{T_2}\underline{\epsilon}, \alpha^{T_2}\overline{\epsilon}) \ \forall i_2 \in \Kcal_2$, we get $|\Scal_X(t'+T_2, \alpha^{T_2}\underline{\epsilon},\alpha^{T_2}\overline{\epsilon})|<|\Scal_X(t'+T_1, \alpha^{T_1}\underline{\epsilon},\alpha^{T_1}\overline{\epsilon})|$.

This logic of the existence of finite time $T_p\geq T_{p-1}+1$ such that $|\Scal_X(t'+T_p, \alpha^{T_p}\underline{\epsilon},\alpha^{T_p}\overline{\epsilon})|<|\Scal_X(t'+T_{p-1}, \alpha^{T_{p-1}}\underline{\epsilon},\alpha^{T_{p-1}}\overline{\epsilon})|$ can be continued iteratively for $p\geq 2$. We first define $\Kcal_p =\{i_p \in \Fcal_{\Rcal}[t'+T_p]\mid x_{i_p}[t'+T_p] \in [\overline m[t']+\alpha^{T_{p}}\underline{\epsilon}, \overline M[t']-\alpha^{T_{p}}\overline{\epsilon}]\}$. Furthermore, we define $\Dcal_{i_p}[t'+t_p] = \Ncal_{i_p}[t'+t_p]\cap\Fcal_{\Rcal}[t'+t_p] \setminus \Scal_X(t'+t_p, \alpha^{t_p}\underline{\epsilon},\alpha^{t_p}\overline{\epsilon})$ for any $t_p \in [T_{p-1},T_p-1]$. Let $T_p\geq T_{p-1}+1$ be the earliest time where $\Kcal_{p-1}\subseteq \Kcal_p$. Using the prior arguments, we conclude that $x_{i_{p-1}}[t'+T_p] \in [\overline m[t'] +\alpha^{T_p} \underline{\epsilon}, \overline M[t'] -\alpha^{T_p} \overline{\epsilon}]$ $\forall i_{p-1} \in \Kcal_{p-1}$, and thus $\Kcal_{p-1} \subseteq \Kcal_p$. Because all $i_{p-1}\in \Kcal_{p-1} \subseteq \Fcal_{\Ccal}$ must belong to $\Fcal_{\Rcal}[t'+t_p]$ at some $t_p \in [T_{p-1}, T_p-1]$, such finite $T_p$ exists.

Furthermore, using the prior arguments, for any $t_p\in [T_{p-1}, T_p-1]$ where $|\Fcal_{\Rcal}[t'+t_p]|>0$, we can have nonempty
$\Scal_{p,t_p}=\{i_p \in \Scal_X(t'+t_p, \alpha^{t_p}\underline{\epsilon},\alpha^{t_p}\overline{\epsilon}) \mid |\Dcal_{i_p}[t'+t_p]|\geq 2F+1\}$. If it is nonempty, since $\Acal$ is $F$-local, $\Dcal_{i_p}[t'+t_p]$ contains at least $F+1$ normal agents $\forall i_p \in \Scal_{p,t_p}$. Then, by definition of $\Scal_X(t'+t_p, \alpha^{t_p}\underline{\epsilon},\alpha^{t_p}\overline{\epsilon})$, $x_{i_p}[t'+t_p]>x_j^{i_p}[t'+t_p]$ or $x_{i_p}[t'+t_p]<x_j^{i_p}[t'+t_p]$ $\forall j \in \Dcal_{i_p}[t'+t_p]$. This means $i_p$ will use at least one in-neighbor's state within $[\overline{m}[t']+\alpha^{t_p}\underline{\epsilon},\overline{M}[t']-\alpha^{t_p}\overline{\epsilon} ]$ to update $x_{i_p}[t'+t_p+1]$, whose bounds are $x_{i_p}[t'+t_p+1]\in [\overline m[t'] +\alpha^{t_p+1} \underline{\epsilon}, \overline M[t'] -\alpha^{t_p+1} \overline{\epsilon}]$ with the similar reasoning as~\eqref{eq:lower_bounds} and~\eqref{eq:upper_bounds}, which is also contained by $[\overline m[t'] +\alpha^{T_p} \underline{\epsilon}, \overline M[t'] -\alpha^{T_p} \overline{\epsilon}]$. Now denote $\Scal_p = \cup_{t_p=T_{p-1}}^{T_p-1}  \Scal_{p,t_p}$. Then, because (i) $\Scal_p \subseteq \Kcal_p$, (ii) $\Kcal_{p-1} \subseteq \Kcal_p$, (iii) $\Scal_{p}\not\subseteq \Kcal_{p-1}$, and (iv) $i_p \notin \Scal_X(t'+T_p, \alpha^{T_p}\underline{\epsilon}, \alpha^{T_p}\overline{\epsilon}) \ \forall i_p \in \Kcal_p$, we get $|\Scal_X(t'+T_p, \alpha^{T_p}\underline{\epsilon},\alpha^{T_p}\overline{\epsilon})|<|\Scal_X(t'+T_{p-1}, \alpha^{T_{p-1}}\underline{\epsilon},\alpha^{T_{p-1}}\overline{\epsilon})|$.

Since $\Fcal_{\Ccal}$ is finite, there must exist a finite time $\overline{T}\geq 1$ such that $\Scal_X(t'+\overline{T}, \alpha^{\overline{T}}\underline{\epsilon},\alpha^{\overline{T}}\overline{\epsilon})=\emptyset$. Then, $\forall i \in \Fcal_{\Ccal}$,
\eqnN{
\overline m[t']+\alpha^{\overline{T}}\underline{\epsilon}\leq x_i[t'+\overline{T}]\leq \overline M[t']-\alpha^{\overline{T}}\overline{\epsilon}.}

Considering $V[t]=\overline M[t]-\overline m[t]$, we get:
\eqn{V[t'+\overline{T}]&=\overline M[t'+\overline{T}]-\overline m[t'+\overline{T}]  \nonumber\\
& \leq \overline M[t']-\alpha^{\overline{T}}\overline{\epsilon} - \left(\overline m[t']+\alpha^{\overline{T}}\underline{\epsilon}\right)\nonumber \\
& = V[t'] - \alpha^{\overline{T}}\left(\overline{\epsilon}+\underline{\epsilon}\right) \label{eq:almost_there}.}

Using the fact that $\underline{\epsilon}=C_r-\overline m[t']$ and $\overline{\epsilon}=\overline M[t']-C_r$, from~\eqref{eq:almost_there}, we get
\eqnN{V[t'+\overline{T}]
 &\leq V[t']  - \alpha^{\overline{T}}V[t']\\
 & = \left(1-\alpha^{\overline{T}}\right)V[t'].
}
The above analysis can be repeated to show that 
\eqnN{V[t'+(\sigma+1)\overline{T}] \leq \left(1-\alpha^{\overline{T}}\right)V[t'+\sigma\overline{T}],
}
for $\sigma\in \mathbb Z_{\geq 0}$, which can be expressed as:
\eqn{V[t'+(\sigma+1)\overline{T}] \leq \left(1-\alpha^{\overline{T}}\right)^{(\sigma +1)}V[t'].
\label{eq:convergence_equation}
}
Since $\left(1-\alpha^{\overline{T}}\right)<1$, $V[t]$ converges to zero as $t=t'+(\sigma+1)\overline{T}\to \infty$, which completes the proof.
\end{proof}

\begin{remark}
The proof of~\Cref{thm:main} follows the similar approach of~\cite[Theorem 1]{time_varying_strongly_usevitch20}, but with a key distinction. While~\cite[Theorem 1]{time_varying_strongly_usevitch20} assumes that the \textit{entire} graph is strongly $(2F+1)$-robust with respect to $\Lcal$ over a fixed time period, our analysis assumes that some \textit{individual followers} become active and thus belong to a strongly $(2F+1)$-robust subgraph infinitely often.
\end{remark}
\Cref{thm:main} shows that any follower that is active (and thus is part of a strongly $(2F+1)$-robust subgraph) infinitely often achieves partial leader-follower consensus. Intuitively, convergence requires continual state updates while being sufficiently connected to filter out adversarial values —-- a condition ensured by the BP-MSR algorithm only for followers who are active infinitely often. This allows us to characterize the convergent set without assuming robustness of the entire network. Note, however, that belonging to a strongly $(2F+1)$-robust subgraph infinitely often is necessary but not sufficient for convergence, since adversaries can indirectly manipulate normal followers' activation states by misreporting their own states (as discussed in~\Cref{remark:on_lemma}).

Finally, the BP-MSR algorithm also guarantees safety for all non-convergent followers by ensuring that their states remain within the convex hull of normal agents whenever $f_r[t]$ is constant. This is because followers transmit and update their states only when they are active. Therefore, in the worst case, all followers may be non-convergent (i.e., $\Fcal_{\Ccal}=\emptyset$), but their states remain bounded.

Now we show that the BP-MSR algorithm also ensures full resilient leader-follower consensus under certain conditions:

\begin{corollary}
Let all the conditions in~\Cref{thm:main} hold. If $\Fcal_{\Ccal}=\Fcal_{\Mcal}$, followers in $\Fcal_{\Mcal}$ will achieve full resilient leader-follower consensus.
\end{corollary}
\begin{remark}
The corollary shows that if there exists an infinite sequence of finite time instants $\{\tau_i\}_{i\in \Zbb_{\geq 0}}$ such that $\bigcup_{t=\tau_{i}}^{\tau_{i+1}} \Gcal[t]$ has all normal followers activated via~\eqref{eq:percolation} (which implies that the union of graph is strongly $(2F+1)$-robust), then the normal followers can achieve full resilient leader-follower consensus. This result in a way generalizes~\cite[Theorem 1]{time_varying_strongly_usevitch20}, which requires the entire network to be strongly $(2F+1)$-robust over every fixed window of $T$ time steps.

\end{remark}

\subsection{Convergent Set Analysis}
\label{sec:convergent_set}
While~\Cref{thm:main} gives conditions for convergence, determining the exact convergent set $\Fcal_{\Ccal}$ as defined in~\Cref{thm:main} a priori is challenging, as it may depend on adversaries' activation states. We therefore provide a more concrete characterization of $\Fcal_{\Ccal}$ by identifying its subset and superset.

\begin{lemma}
Let all conditions in~\Cref{thm:main} hold. Define $N_1$ and $N_2$ as the numbers of convergent followers when all adversaries $a \in \Acal$ share $q_a^{u}[t,k]=0$ and $q_a^{u}[t,k]=1$, respectively, with their out-neighbors $u\in \Ncal_a^o[t]$, for all $k \in\{0,\dots,f-1\}$ and $t \in \Zbb_{\geq 0}$. Then $N_1\leq |\Fcal_{\Ccal}|\leq N_2$. 
\label{lem:extreme_case}
\end{lemma}
\begin{proof}
Let $\Fcal_{\Rcal}[t] =  \{i \in \Fcal_{\Mcal} \mid q_i[t]=1\}$. Note that $i \in \Mcal$ strictly follows~\eqref{eq:percolation} to update its activation state $q_i[t,k]$ for all iterations. By~\Cref{assum:misreport}, we know $q_a^{u}[t,k]\in\{0,1\}$ for any $a \in \Acal$. Then, for a given $\Gbb$, $|\Fcal_{\Rcal}[t]|$ is determined by whether $a \in \Acal$ shares $q_a^{u}[t,k]=1$ or $q_a^{u}[t,k]=0$ in each $k\in\{0,\dots,f-1\}$. By the protocol~\eqref{eq:percolation}, for any $i \in \Fcal_{\Mcal}$, $q_i[t]=1$ if there exists $\bar k\in\{0,\dots,f-1\}$ such that $\sum_{j \in \Ncal_i[t]\setminus \Acal} q_j^i[t,\bar k] + \sum_{a \in \Acal} q_a^i[t,\bar k] \geq 2F+1$. Therefore, $|\Fcal_{\Rcal}[t]|$ at time $t$ is minimized when all adversaries share $q_a^{u}[t,k]=0$ for every $k\in\{0,\dots,f-1\}$ and maximized when they share $q_a^{u}[t,k]=1$ for every $k$. Hence, all $a\in \Acal$ sharing $q_a^{u}[t,k]=0$ and $q_a^{u}[t,k]=1$ $\forall  k\in\{0,\dots,f-1\}$ and $\forall  t \in \Zbb_{\geq 0}$ minimize and maximize $|\Fcal_{\Ccal}|$, respectively. 
\end{proof}


\Cref{lem:extreme_case} characterizes the two extreme scenarios in which the convergent set is minimized or maximized, based on how adversaries behave. Using this lemma, we provide a constructive method to identify its subset and superset:
\begin{prop} 
\label{prop:subset_superset}Let all conditions in~\Cref{thm:main} hold. For each time $t$, let $\Vcal^1[t]\subseteq \Mcal$ and $\Vcal^2[t]\subseteq \Vcal$ denote two different sets of nodes such that the induced subgraphs $\Gcal^1[t] = (\Vcal^1[t], \Ecal^1[t])$ and $\Gcal^2[t] = (\Vcal^2[t], \Ecal^2[t])$ are strongly $(2F+1)$-robust with respect to $\Lcal_{\Mcal}$ and $\Lcal\cup \Acal$, respectively. For each node $i \in \Vcal$, let $\Tcal^j_i= \{ t \geq t_C \mid i \in \Vcal^j[t] \}$, and let $\Fcal_{\Ccal}^j=\{z \in \Fcal_{\Mcal} \mid |\Tcal_z^j|=\infty\}$, $j=1,2$. Then, $\Fcal_{\Ccal}^1 \subseteq \Fcal_{\Ccal}\subseteq \Fcal_{\Ccal}^2$. 
\end{prop}

\begin{proof}

We know that $\Gcal^1[t]$ is strongly $(2F+1)$-robust with respect to $\Lcal_{\Mcal}$ and $\Vcal^1[t] \subseteq \Mcal$. 
Hence, by~\Cref{lem:strong_r_not_mine}, every $i \in \Vcal^1[t] \cap \Fcal_{\Mcal}$ has $\bar{k}\in\{0,\dots, f-1\}$ such that $\sum_{j \in \Ncal_i[t] \cap \Mcal} q_j^i[t,\bar{k}] \geq 2F+1$.
Thus, if all adversaries share $q_a^{u}[t,k] = 0$ for all $k\in\{0,\dots,f-1\}$, only the followers in $\Vcal^1[t] \cap \Fcal_{\Mcal}$ are guaranteed to have $q_i[t]=1$ at time $t$. It follows that any $i$ with $|\Tcal_i^1| = \infty$ has $\lim\sup_{t\to \infty} q_i[t]=1$, and thus convergent by~\Cref{thm:main}. 
By~\Cref{lem:extreme_case}, this defines the smallest possible convergent set, yielding $\Fcal_{\Ccal}^1 \subseteq \Fcal_{\Ccal}$.

An analogous argument applies to $\Gcal^2[t]$, which is strongly $(2F+1)$-robust with respect to $\Lcal \cup \Acal$. 
If all adversaries share $q_a^{u}[t,k] = 1$, we can consider all adversaries as leaders. Then by~\Cref{lem:strong_r_not_mine}, $q_i[t]=1$ for all $i \in \Vcal^2[t] \cap \Fcal_{\Mcal}$ at time $t$. 
Thus, any $i$ with $|\Tcal_i^2| = \infty$ has $\lim\sup_{t\to \infty} q_i[t]=1$ and convergent by~\Cref{thm:main}. 
By~\Cref{lem:extreme_case}, this defines the largest possible convergent set, yielding $\Fcal_{\Ccal} \subseteq \Fcal_{\Ccal}^2$.
\end{proof}
Through~\Cref{prop:subset_superset}, we can explicitly determine the subset and superset of convergent set $\Fcal_{\Ccal}$. The subset $\Fcal_{\Ccal}^1$ is a set of followers that belong infinitely often to a subgraph that is induced \textit{only by normal agents} and strongly $(2F+1)$-robust with respect to $\Lcal_{\Mcal}$. Conversely, the superset $\Fcal_{\Ccal}^2$ consists of followers that belong infinitely often to a subgraph that is strongly $(2F+1)$-robust with respect to $\Lcal\cup \Acal$. Examples of such subgraphs are shown in~\Cref{fig:sequence}.

\section{Simulation Results}
\label{sec:sim}
\begin{figure}
    \centering
    \includegraphics[width=\linewidth]{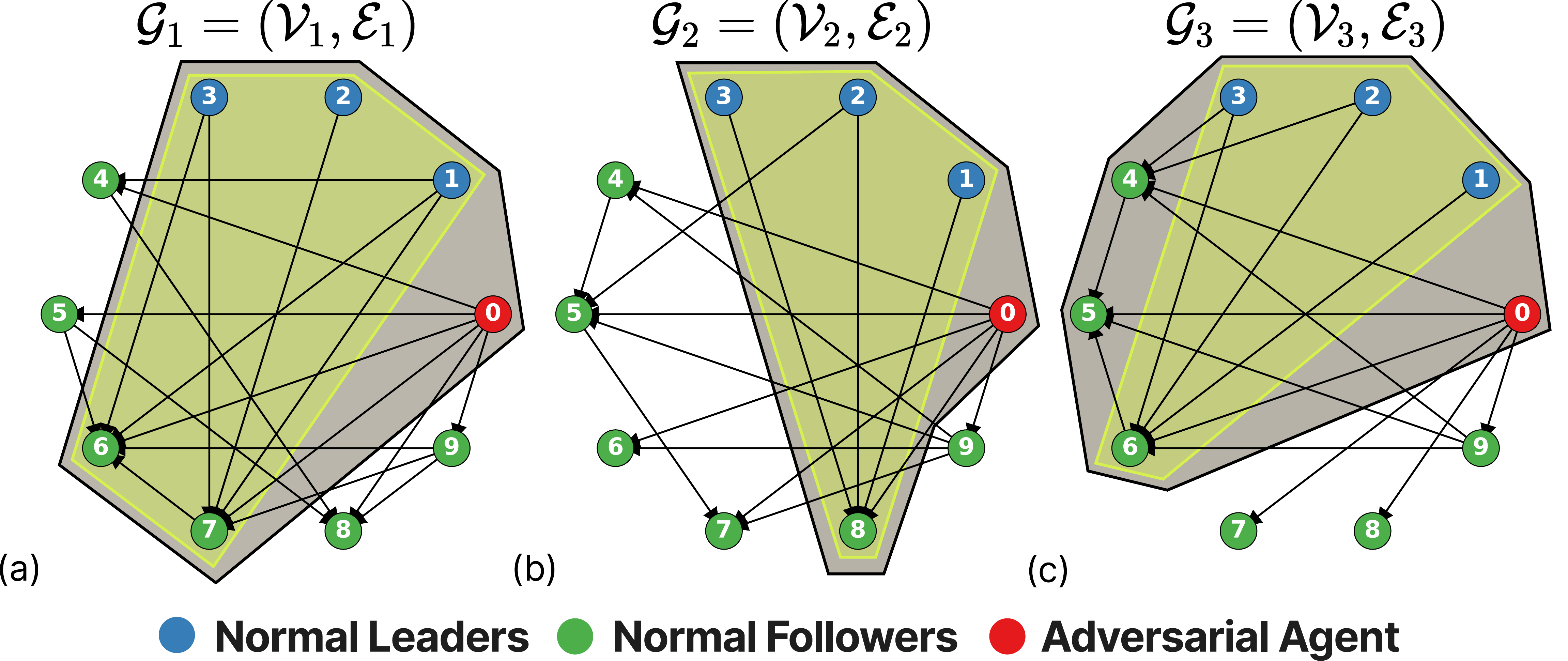}
    \caption{Three digraphs where blue, green, and red nodes represent normal leaders, normal followers, and adversaries, respectively, each under a $1$-local attack. Yellow-shaded subgraphs are strongly $3$-robust with respect to $\mathcal{L}_{\mathcal{M}}$, while the gray-shaded subgraph is strongly $3$-robust with respect to $\mathcal{L} \cup \mathcal{A}$.}
    \label{fig:sequence}
\end{figure}

We evaluate the effectiveness of our method through two sets of simulations. First, we illustrate~\Cref{thm:main} by comparing the performance of the BP-MSR algorithm against those of other existing resilient consensus algorithms. Second, we empirically validate the set bounds in~\Cref{prop:subset_superset}.
For all simulations, we use the graphs in~\Cref{fig:sequence}, each with leaders $\Lcal=\Lcal_{\Mcal}= \{1,2,3\}$ and a $1$-local $\Acal = \{0\}$. The adversary shares $x_0^j[t] = 1000 \cdot \sin\big((t+j)/5\big)$ with the normal followers $j \in \{4,5,6,7,8,9\}$. The initial states of followers are generated randomly in the interval $[-1000, 1000]$. Leaders update their states using $f_r$, drawing the same random value in $[-1000, 1000]$ every $50$ time steps. Code is available here.\footnote{\href{https://github.com/joonlee16/partial-leader-follower-consensus}{https://github.com/joonlee16/partial-leader-follower-consensus}}

\subsubsection{Comparison} For comparison, we include the W-MSR algorithm with parameter $F$, which guarantees resilient leader-follower consensus with $F$-local adversary set $\Acal$ when the network is always either strongly $(2F+1)$-robust or jointly $(F+1)$-robust with $1$ hop~\cite{LeBlanc13, strongly_usevitch18, yuan2024reaching}. We also consider the SW-MSR algorithm with parameters $T$ and $F$, which ensures consensus when the network satisfies strong $(T, t_0, 2F+1)$-robustness~\cite{time_varying_strongly_usevitch20}. To evaluate the applicability of these guarantees, we analyze the robustness of $\Gcal_1$, $\Gcal_2$, and $\Gcal_3$ from~\Cref{fig:sequence}. Each graph is at most strongly $1$-robust with respect to $\Lcal$ and jointly $1$-robust with one hop, implying that resilient consensus cannot be guaranteed by W-MSR~\cite{time_varying_strongly_usevitch20, yuan2024reaching}. 
Similarly, any periodic repetition of the graphs in~\Cref{fig:sequence} fails to satisfy strong $(3,0,3)$-robustness, so consensus is not guaranteed under SW-MSR either~\cite{time_varying_strongly_usevitch20}.

\begin{figure}
\centering
\includegraphics[width=\linewidth]{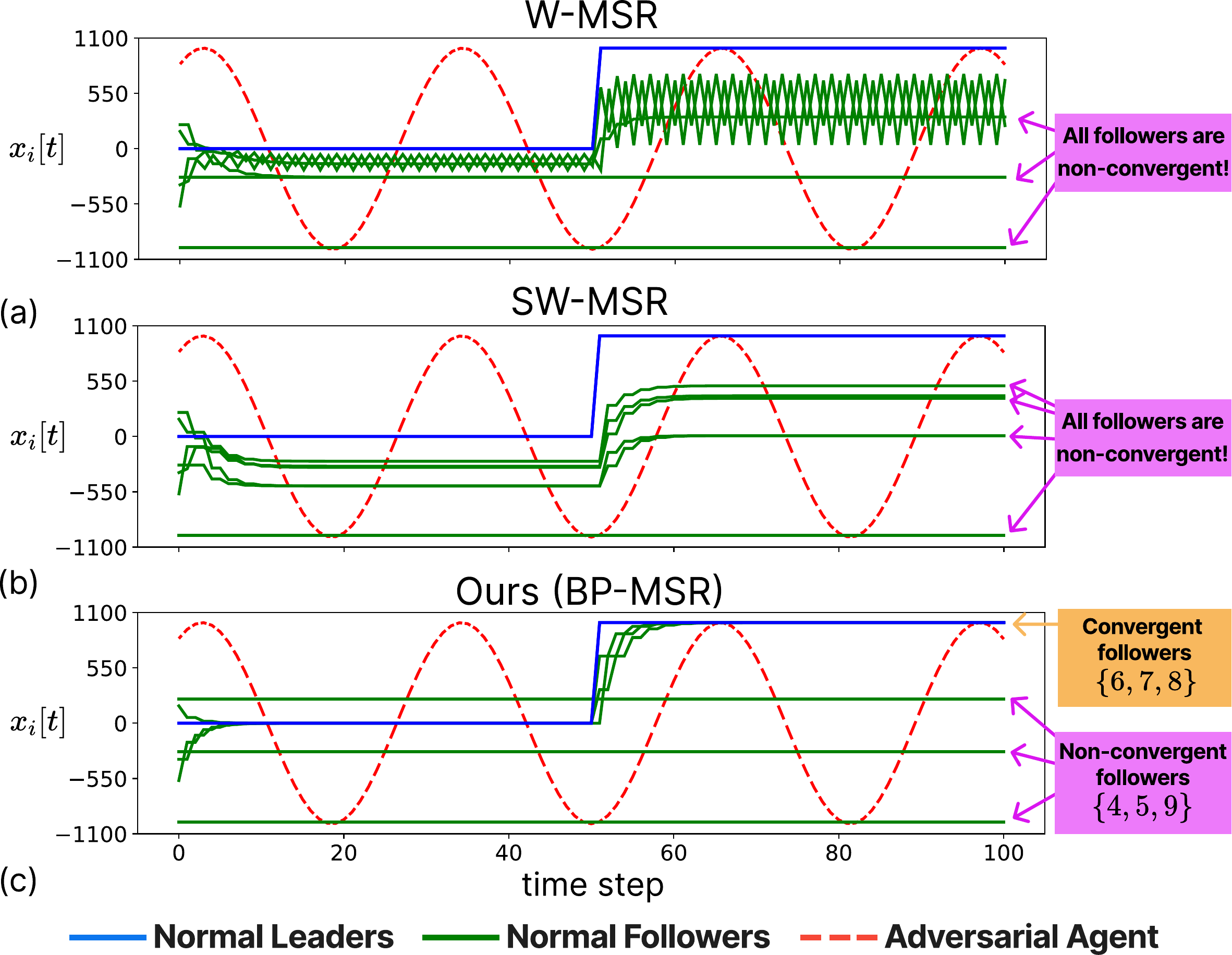}
\caption{Consensus performance of (a) W-MSR, (b) SW-MSR ($T=2$), and (c) BP-MSR algorithms, all with $F=1$, under the graph sequence 
$\mathbb{G}=(\mathcal{G}[t])_{t\in \mathbb{Z}_{\geq 0}}$, where $\mathcal{G}[2\tau]=\mathcal{G}_1$ and $\mathcal{G}[2\tau+1]=\mathcal{G}_2$, for $\tau\in\mathbb{Z}_{\geq 0}$ (see \Cref{fig:sequence} (a)-(b) for $\Gcal_1$ and $\Gcal_2$). Partial leader-follower consensus is achieved by followers $\{6,7,8\}$ only via the BP-MSR algorithm. For clarity, only the adversary state received by follower $5$ is plotted.}
    \label{fig:comparison}
\end{figure}

\Cref{fig:comparison} shows results consistent with the analysis above and~\Cref{thm:main}. 
We consider three cases where each normal agent runs the W-MSR, SW-MSR (with $T=2$), and BP-MSR algorithms, all with $F=1$, under the sequence $\Gbb=(\Gcal[t])_{t\in \Zbb_{\geq 0}}$, where $\Gcal[2\tau]=\Gcal_1$ and $\Gcal[2\tau+1]=\Gcal_2$ for all $\tau \in \Zbb_{\geq 0}$. 
Neither the W-MSR nor the SW-MSR algorithm enables the followers to track the leader’s reference state (blue). 
In contrast, followers in $\Fcal_{\Ccal}=\{6,7,8\}$ reach partial leader-follower consensus under the BP-MSR algorithm, as each satisfies $\limsup_{t \to \infty} q_i[t] = 1$ and thus belongs to strongly $3$-robust subgraphs infinitely often (see yellow-shaded regions in~\Cref{fig:sequence} (a)-(b)). Meanwhile, followers $\{4,5,9\}$ are non-convergent under BP-MSR algorithm but remain within the convex hull of the normal agents’ states.

\begin{figure}
    \centering
\includegraphics[width=\linewidth]{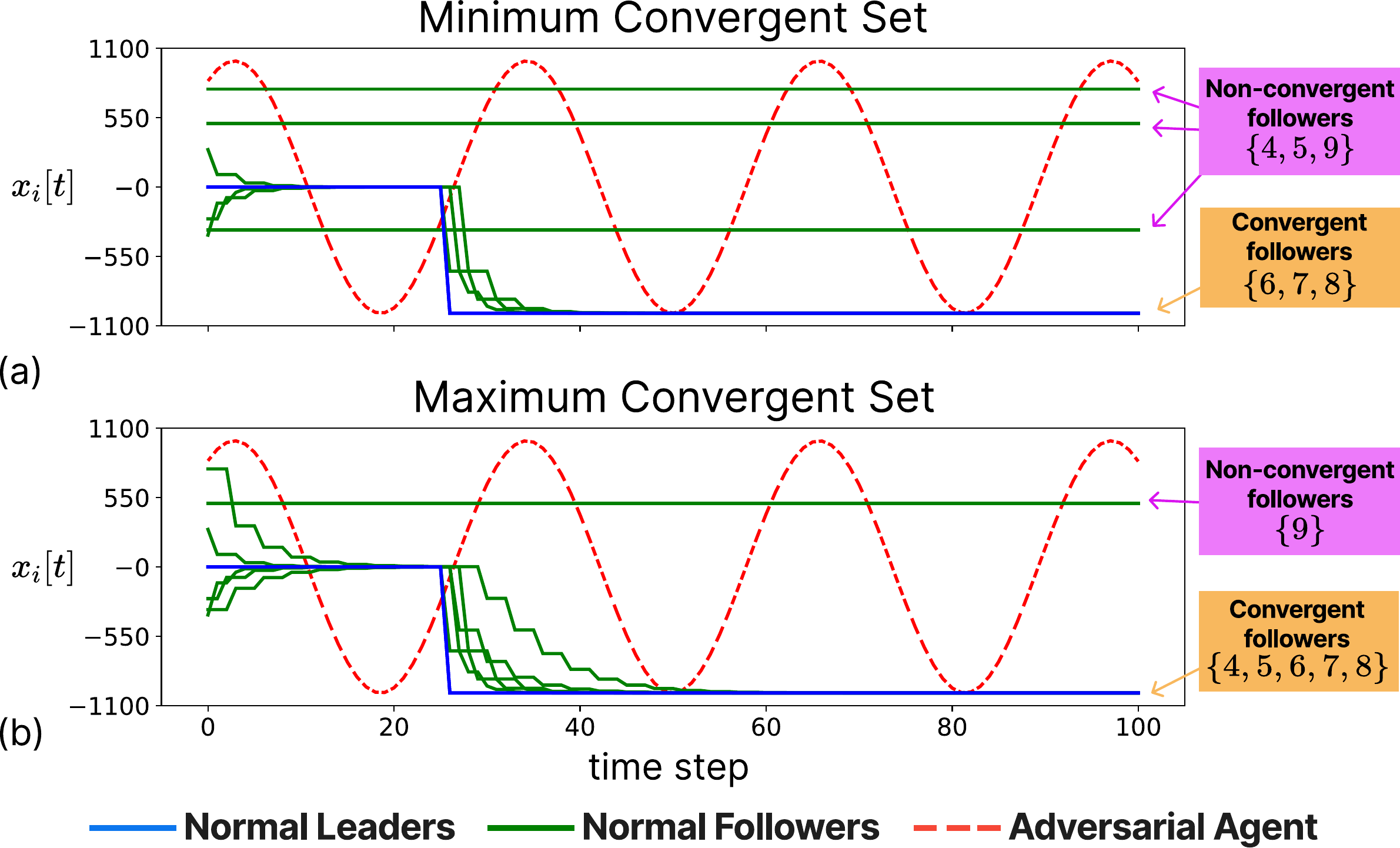}
    \caption{Consensus performance of the BP-MSR algorithm under the periodic graph sequence $\mathbb{G}$, where $\mathcal{G}_1$, $\mathcal{G}_2$, and $\mathcal{G}_3$ from~\Cref{fig:sequence} repeat periodically. The adversary $0$ (a) always shares $q_0^u[t,k]=0$ for all time with out-neighbors $u\in \Ncal_0^o[t]$, yielding the smallest $\Fcal_{\Ccal}$ possible, and (b) always shares $q_0^u[t,k]=1$ for all time, yielding the largest $\Fcal_{\Ccal}$ possible. For clarity, we plot only the adversary’s state received by follower~5.}
    \label{fig:sub_super}
\end{figure}

\subsubsection{Convergent Set Analysis}\label{sec:convergent} Now we analyze the convergent set $\Fcal_{\Ccal}$ under the BP-MSR algorithm with the graph sequence $\Gbb$ where $\Gcal_1$, $\Gcal_2$, and $\Gcal_3$ repeat periodically. Although the exact $\Fcal_{\Ccal}$ might be difficult, we can identify its subset and superset using~\Cref{prop:subset_superset}. In $\mathcal{G}_1$, $\mathcal{G}_2$, and $\mathcal{G}_3$, the followers $\{6,7\}$, $\{8\}$, and $\{6\}$ respectively belong to strongly $3$-robust subgraphs with respect to $\mathcal{L}_{\mathcal{M}}$ (highlighted in yellow in~\Cref{fig:sequence}). Similarly, in $\mathcal{G}_1$, $\mathcal{G}_2$, and $\mathcal{G}_3$, the followers $\{6,7\}$, $\{8\}$, and $\{4,5,6\}$ respectively belong to strongly $3$-robust subgraphs with respect to $\Lcal \cup \Acal$ (as highlighted in gray in~\Cref{fig:sequence}). Since the graphs are repeating periodically, we can deduce $\{6,7,8\} \subseteq \Fcal_{\Ccal} \subseteq \{4,5,6,7,8\}$ using~\Cref{prop:subset_superset}. This aligns with~\Cref{fig:sub_super} (a) and (b), which illustrate the smallest and largest convergent sets (by~\Cref{lem:extreme_case}), obtained when Byzantine agent $0$ always shares $q_0^u[t,k]=0$ and $q_0^u[t,k]=1$ with its out-neighbors, respectively. In case (a), only followers $\{6,7,8\}$ converge, whereas in case (b), $\{4,5,6,7,8\}$ converge.

\section{Conclusions}
\label{sec:conc}
  
This paper studies resilient leader-follower consensus for a subset of followers in an arbitrary sequence of time-varying graphs. We propose a novel distributed algorithm that allows each follower to locally decide when to share and update its state based on its connectivity at each time step. We provide a theoretical characterization of the followers guaranteed to achieve leader-follower consensus. Finally, we support our results with simulations showing that our method allows a subset of followers to achieve consensus even when robustness conditions of the entire network are not satisfied.






\bibliographystyle{IEEEtran}
\bibliography{references_ll}

\end{document}

%% file: preamble.tex
\usepackage{cite}
\usepackage{amsmath,amssymb,amsfonts}
\usepackage{algorithm}
\usepackage{algpseudocode}
\usepackage{graphicx}
\usepackage{textcomp}
\usepackage{setspace}
\usepackage{booktabs}

\usepackage{epsfig}
\usepackage{times} 
\usepackage{svg}

\usepackage[nolist, nohyperlinks]{acronym}

\usepackage{bm}
\usepackage{breqn}
\usepackage{amsmath}
\usepackage{amssymb}
\usepackage{tabularx}
\usepackage[dvipsnames]{xcolor}

\usepackage{amsthm}

\newcommand{\Gbb}{\mathbb{G}}

\newcommand{\Zbb}{\mathbb{Z}}

\newcommand{\Acal}{\mathcal{A}}
\newcommand{\Bcal}{\mathcal{B}}
\newcommand{\Ccal}{\mathcal{C}}
\newcommand{\Dcal}{\mathcal{D}}
\newcommand{\Ecal}{\mathcal{E}}
\newcommand{\Fcal}{\mathcal{F}}
\newcommand{\Gcal}{\mathcal{G}}
\newcommand{\Hcal}{\mathcal{H}}
\newcommand{\Ical}{\mathcal{I}}

\newcommand{\Kcal}{\mathcal{K}}
\newcommand{\Lcal}{\mathcal{L}}
\newcommand{\Mcal}{\mathcal{M}}
\newcommand{\Ncal}{\mathcal{N}}
\newcommand{\Ocal}{\mathcal{O}}

\newcommand{\Rcal}{\mathcal{R}}
\newcommand{\Scal}{\mathcal{S}}
\newcommand{\Tcal}{\mathcal{T}}

\newcommand{\Vcal}{\mathcal{V}}
\newcommand{\Wcal}{\mathcal{W}}
\newcommand{\Xcal}{\mathcal{X}}

\newcommand{\eqn}[1]{\begin{align} #1 \end{align}}
\newcommand{\eqnN}[1]{\begin{align*} #1 \end{align*}}

\newcommand{\norm}[1]{\left\Vert #1 \right \Vert}

\theoremstyle{plain}
\newtheorem{theorem}{Theorem}
\newtheorem{corollary}{Corollary}
\newtheorem{lemma}{Lemma}
\newtheorem{problem}{Problem}
\newtheorem{definition}{Definition}
\newtheorem{prop}{Proposition}

\theoremstyle{definition}
\newtheorem{assumption}{Assumption}
\newtheorem{remark}{Remark}

\theoremstyle{remark}

\makeatletter
\let\NAT@parse\undefined
\makeatother
\usepackage{hyperref}
\hypersetup{
colorlinks, 
    linkcolor={red!50!black},
    citecolor={blue!50!black},
    urlcolor={blue!80!black}
}
\usepackage{cleveref}

\crefformat{problem}{Problem~#2#1#3}
\crefformat{assumption}{Assumption~#2#1#3}
\crefformat{prop}{Proposition~#2#1#3}